\documentclass[11pt]{article}
\usepackage[margin=1in]{geometry}
\usepackage{amsmath, amssymb, amsthm}
\usepackage{mathtools}
\usepackage{algorithm}
\usepackage{algpseudocode}
\usepackage{hyperref}
\usepackage{graphicx}
\usepackage{cite}
 
\theoremstyle{plain}
\newtheorem{theorem}{Theorem}
\newtheorem{conjecture}{Conjecture}[section]
\newtheorem{lemma}[theorem]{Lemma}

\newtheorem{corollary}[theorem]{Corollary}
 
\theoremstyle{definition}
\newtheorem{definition}[theorem]{Definition}
\newtheorem{remark}[theorem]{Remark}

\DeclareMathOperator{\poly}{poly}
 
\title{Adversarial Configurations for the\\ReCom Transition Function}
\author{Micah Gold\\micahcgold@gmail.com}
\date{\today}
 
\begin{document}
\maketitle
 
\begin{abstract}
\textsf{ReCom} is a leading Markov Chain Monte Carlo algorithm for sampling balanced graph partitions in computational redistricting. At each step, its transition function proposes a new partition by merging two adjacent districts and if possible re-splitting the conjoined region. The transition function is efficient in practice, however, it is unknown whether it is guaranteed to run in polynomial time. In this report we exhibit an explicit family of $3$-partitions on planar square grid graphs from which \textsf{ReCom} requires an exponentially large expected number of steps to re-split the graph (even if we admit approximately balanced splits), showing that in the worst case \textsf{ReCom} does not run in polynomial time. Notably, this result implies that ReCom is not \emph{technically} rapidly mixing (if started from an adversarial configuration, ReCom requires exponentially many steps to reach the stationary distribution). 
\end{abstract}
 
\section{Introduction}
 
In the United States, members of the House of Representatives are elected to represent geographic districts. States are divided into districts through a map-drawing process known as redistricting. Redistricting is politically important. Shifting the boundaries of a district by even a few miles can alter the composition of its voter-base and sway election outcomes. Partisan gerrymandering, or strategically drawing districts to maximize political advantage, has haunted redistricting since the 19th century~\cite{masshist1812}. In recent years, states from Utah to Michigan have passed laws prohibiting partisan gerrymandering~\cite{campaignlegal2018} and courts have increasingly been asked to adjudicate whether proposed maps are fair ~\cite{brennancenter2021}. Unfortunately, quantifying unfairness is non-trivial. Fortunately, experts have converged on an intuitive and theoretically appealing approach: the \emph{ensemble method}. 

The ensemble method entails drawing a large sample of ``reasonable'' redistricting plans from an explicit distribution, and checking whether the proposed plan is a statistical outlier. The leading algorithm in this space is \textsf{ReCom} ~\cite{deford2021}, together with its provably reversible variant  \textsf{Reversible ReCom} \cite{cannon2022}.
 
\subsection{ReCom and Problem 2.8}

\textsf{ReCom} is a Markov Chain Monte Carlo algorithm used to sample redistricting plans \cite{deford2021}. Sampling from the \textsf{ReCom} steady-state approximates sampling from the distribution of spanning trees over the graph of census blocks (in the ``Reversible" variant of the algorithm used in this paper, which adds a Metropolis-Hastings acceptance probability for state transitions, the steady-state exactly converges to the spanning tree distribution \cite{cannon2022}). As a result, districts sampled from \textsf{ReCom} have nice properties. They are always contiguous (an effective requirement for real-world districts) and usually compact. 

\textsf{ReCom} maintains a balanced $k$-partition $\mathcal{P} = \{P_1, \dots, P_k\}$ of a planar graph $G$ whose vertices represent census blocks. At each step, the chain (i) selects two random pieces $P_i, P_j$, (ii) if they are adjacent, samples a uniform spanning tree $T$ of $P_i \cup P_j$, (iii) checks whether some edge $e \in T$ splits $T$ into two approximately balanced halves, and if so accepts a swap with a reversibility correction probability. Otherwise, the chain returns $\mathcal{P}$ unchanged. 
The full pseudocode for \textsf{Reversible ReCom} appears as Algorithm~\ref{alg:recom} in Appendix~\ref{app:recom}.

Empirically, \textsf{ReCom} mixes quickly and produces high-quality samples \cite{deford2021,cannon2022}, yet many of its theoretical properties remain unknown. Chief among them is whether a single step always runs in polynomial time. The bottleneck is (iii). \textsf{ReCom} only makes progress when the spanning tree sampled from the merged region $P_i \cup P_j$ happens to have a balanced cut-edge, so the running time is governed by how likely a uniform spanning tree of that region is to be splittable. If this probability is inverse-polynomial the step is efficient; if it is exponentially small the chain stalls on that move.

Charikar, Liu, Liu, and Vuong \cite{charikar2022} have conjectured that, on grids, splittable partitions are abundant under the spanning tree distribution:
\begin{conjecture}[Charikar, Liu, Liu, and Vuong \cite{charikar2022}]
For the $m \times n$ grid graph, the proportion of balanced $k$-partitions under the spanning tree distribution is at least $1/\poly(m,n)$ when $k = O(1)$.
\end{conjecture}
Cannon, Pegden, and Tucker-Foltz \cite{cannon2024} made the link to (iii) precise, showing that the proportion of balanced partitions is equivalent, up to polynomial factors, to the probability that a single uniform spanning tree can be split into balanced subtrees. They then proved this probability is inverse-polynomial for grids and a broad class of grid-like graphs. In other words, on a well-shaped region, a \textsf{ReCom} step is likely to succeed.

But \textsf{ReCom} does not get to choose well-shaped regions. It must merge whatever pieces the current partition provides. The positive result therefore leaves open whether the chain can be forced into a state where every available move is bad. Concretely, to stall the chain one would need a balanced partition in which \emph{every} pair of mergeable districts forms a region whose spanning trees are rarely splittable, so that no merge is likely to admit a balanced cut. 

Tucker-Foltz formalized the question as Problem~2.8 on the AIM open problem list for sampling connected balanced graph partitions:
\begin{quote}
\textbf{Problem 2.8.} Does the \textsf{ReCom} transition function run in polynomial time? Or can it get stuck in configurations where it is very unlikely a random tree will be splittable? Can we at least show that such configurations are exponentially unlikely? \cite{aimpl}
\end{quote}
In this paper we address the existence half of Problem~2.8. We answer the question ``do such stuck configurations exist?"
\subsection{Our Result}
Surprisingly, the answer is a resounding yes, even when \textsf{ReCom} is run on planar square grid graphs and even when we allow approximately balanced splits (see Theorem~\ref{thm:main}).

Informally, there exists an infinite family of $3$-partitions $\mathcal{P}_n = \{R_1^n, R_2^n, R_3^n\}$ of planar square grid graphs $G_n$ such that:
\begin{itemize}
    \item Each $R_i^n$ is connected and approximately population-balanced.
    \item The probability that the \textsf{ReCom} transition function finds an approximately balanced cut is exponentially small with respect to the dimensions of the grid.
\end{itemize}
Consequently, the expected number of steps required for \textsf{ReCom} to leave the state $\mathcal{P}_n$ is exponentially large with respect to the dimensions of the grid.

Notably, the existence of such an adversarial configuration means that  \emph{ReCom is not technically rapidly mixing} because there exists a starting configuration from which it cannot reach steady-state in a polynomial number of steps.
\subsection{Argument Outline}
 
The argument proceeds in two stages.
 
\begin{enumerate}
    \item First, we identify a family of graphs we call \emph{highway graphs} -- two large dense subgraphs connected by long $2 \times m$ ladders, with a third dense subgraph tethered to the middle of the highway by a thin tendril --- and prove that any spanning tree of a highway graph that contains two or more vertical rungs on each half of the highway admits no approximately balanced cut (Lemma~\ref{lem:two-rungs-bad}). Second, we use the effective-resistance machinery of \cite{procaccia2022} to show that a uniform spanning tree of the highway contains at least two rungs on each side of the highway except with exponentially small probability (Lemma~\ref{lem:few-rungs-rare}). We combine these two lemmas to show that any single highway graph is exponentially hard to cut in a balanced way.
    \item  We then exhibit a symmetrical 3-partitioned planar graph on which all possible merges produce a highway graph. Finally we demonstrate that \textsf{Reversible ReCom} will take an expected exponential number of steps in order to find a balanced cut and transition to its next state.
\end{enumerate}

\section{Preliminaries}
 
\subsection{Graph Partitions and Reversible ReCom}
 
We consider connected, undirected planar graphs $G = (V, E)$ in which every vertex represents a census unit of unit population. For a positive integer $k$ dividing $|V|$, a \emph{$k$-partition} of $G$ is a collection $\mathcal{P} = \{P_1, \dots, P_k\}$ of pairwise-disjoint vertex subsets covering $V$ such that each induced subgraph $G[P_i]$ is connected. We say $\mathcal{P}$ is \emph{$\varepsilon$-balanced} if
\[
\bigl| |P_i| - |V|/k \bigr| \leq \varepsilon |V|/k \qquad \text{for all } i \in \{1, \dots, k\}.
\]
Two pieces $P_i$ and $P_j$ are \emph{adjacent} if some edge of $G$ has one endpoint in each. 
 
The \textsf{Reversible ReCom} transition function ~\cite{cannon2022} on $\varepsilon$-balanced $k$-partitions of $G$ executes the following procedure at each step:
\begin{enumerate}
    \item Select an unordered pair $\{P_i, P_j\}$ uniformly at random.
    \item If $P_i$ and $P_j$ are not adjacent, return $\mathcal{P}$ unchanged.
    \item Otherwise, sample a uniform spanning tree $T$ of $G[P_i \cup P_j]$.
    \item Search $T$ for an edge $e$ whose removal splits $T$ into two components $P_i', P_j'$, each within $\varepsilon$ tolerance of size $|V|/k$.
    \item If no such $e$ exists, return $\mathcal{P}$ unchanged.
    \item If such an $e$ exists, accept the swap with probability $1/E(P_i', P_j')$, where $E(P_i', P_j')$ is the number of edges of $G$ crossing the cut.
\end{enumerate}
 
We refer to a step that returns $\mathcal{P}$ unchanged as a \emph{stalled step}. Our goal is to construct a state $\mathcal{P}$ from which every step stalls with exponentially high probability.
 
\subsection{Effective Resistance and Spanning Trees}
Borrowing the insight from ~\cite{procaccia2022}, we view an unweighted graph $H$ as an electrical network with a unit resistor on each edge. For any edge $e^* = \{a,b\} \in E(H)$, the \emph{effective resistance} $R_{ab}$ is the voltage required to drive one unit of current from $a$ to $b$ \cite{spielman2019sagt}. Three facts will suffice for our purposes.
 
\begin{lemma}[\cite{procaccia2022}, Lemma 2.1]
\label{lem:resistance-tree}
For any edge $e^* = \{a,b\} \in E(H)$, the effective resistance $R_{ab}$ equals the probability that $e^*$ appears in a uniformly random spanning tree of $H$.
\end{lemma}

\begin{lemma}[\cite{procaccia2022}, Lemma 2.4]
\label{lem:degree-bound}
For any pair of adjacent vertices $a, b$ in $H$ with unit resistors, $R_{ab} \geq 1/\deg_H(a)$.
\end{lemma}
 
We will also borrow the following mechanism for sampling uniform spanning trees, which makes the connection between effective resistance and uniform spanning trees algorithmically explicit \cite{procaccia2022}.
 
\begin{algorithm}[h]
\caption{Sample a uniform spanning tree of $H$ via effective resistance \cite{procaccia2022}.}
\label{alg:ust-resistance}
\begin{algorithmic}[1]
\State $T \gets \emptyset$
\While{$H$ has at least two vertices}
    \State pick any edge $e \in E(H)$
    \State compute $r \gets $ effective resistance of $e$ in $H$
    \State with probability $r$:
        \State \quad add $e$ to $T$ and contract $e$ in $H$
    \State otherwise:
        \State \quad delete $e$ from $H$
\EndWhile
\State \Return $T$
\end{algorithmic}
\end{algorithm}
 
We will use Algorithm~\ref{alg:ust-resistance} in Section~\ref{sec:lemma2} to bound the probability that few rungs appear in $T$. Notably, by Lemma~\ref{lem:resistance-tree}, the probability of a contraction/deletion sequence occurring equals the probability that a uniform spanning tree of the original graph $H$ contains all of the contracted edges and excludes all of the deleted edges. This means that the order in which edges are processed does not affect the joint distribution over outcomes \cite{procaccia2022}.
 
\section{The Highway Graph}
\label{sec:highway}
 
We now describe the central object of our construction: a planar graph $H_m$ whose uniform spanning trees almost never admit an approximately balanced cut.
 
\begin{definition}[Highway graph]
\label{def:highway}
Fix a positive integer $m$ (to be thought of as large) and a balance tolerance $\varepsilon \in [0, \frac{1}{12})$. The \emph{highway graph} $H_m$ is the planar graph on $N$ vertices assembled from the following pieces (see Figure~\ref{fig:highway}):
\begin{enumerate}
    \item \textbf{Three regions.} Three pairwise vertex-disjoint connected planar graphs $R_L$, $R_R$, and $R_T$, each satisfying the size bound
    \[
    |V(R_i)| \;<\; \bigl(\tfrac{1}{2} - \varepsilon\bigr)N - \bigl(2(2m+1) + 1\bigr), \qquad i \in \{L, R, T\}.
    \]
    \item \textbf{The ladder.} A $2 \times (2m+1)$ grid graph $L$, vertex-disjoint from $R_L \cup R_R \cup R_T$. We label its columns $1, 2, \dots, 2m+1$ from left to right, and refer to columns $1, \dots, m$ as the \emph{left half}, columns $m+2, \dots, 2m+1$ as the \emph{right half}, and column $m+1$ as the \emph{central column}. The $2m+1$ vertical edges of $L$ are called \emph{rungs} and form the set~$\mathcal{R}$; the remaining $2 \cdot 2m = 4m$ edges along the top and bottom rows are called \emph{rail edges}. For a spanning tree $T$ of $H_m$ and each side $s \in \{L, R\}$, we write $\mathcal{R}_T^s$ for the set of rungs of $T$ lying in the $s$-half of the ladder.
    \item \textbf{The tendril.} An additional vertex $v_T \notin V(R_L) \cup V(R_R) \cup V(R_T) \cup V(L)$, joined by a single edge to the bottom vertex of column $m+1$ of the ladder and by a single edge to a boundary vertex of $R_T$.
\end{enumerate}
The left end column (column $1$) of the ladder is joined to $R_L$, and the right end column (column $2m+1$) to $R_R$, by a single edge from each of the two end-column vertices to a boundary vertex of the corresponding region. In particular, like every interior ladder vertex, each end-column vertex is incident to exactly one rung and two further edges, so every vertex of the ladder except the bottom vertex of the central column (which additionally carries the tendril) has degree $3$ in $H_m$.
We refer to $L \cup \{v_T\}$ as the \emph{highway corridor}, and we write $\mathcal{R}$ for the set of rungs.
\end{definition}
 
\begin{figure}[h]
    \centering
    \includegraphics[width=0.85\textwidth, trim=0 0 0 80pt, clip]{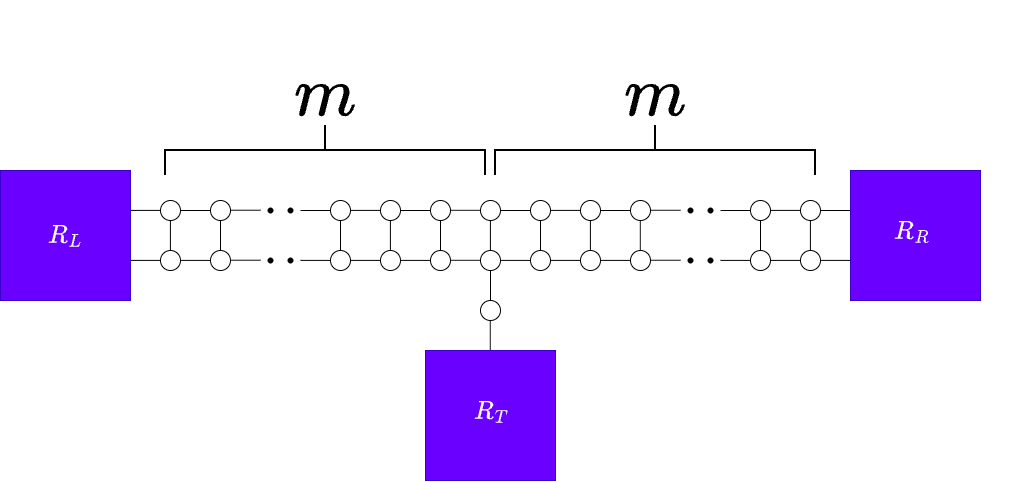}
    \caption{The highway graph $H_m$. Three regions $R_L, R_T, R_R$ are connected by a horizontal $2 \times (2m+1)$ ladder, with a single vertex $v_T$ tendril dropping from the bottom of the central column to $R_T$.}
    \label{fig:highway}
\end{figure}
 
\begin{remark}[Region masses]
\label{rem:block-mass}
The size bound on each region in Definition~\ref{def:highway} is calibrated so that no single region can occupy a balanced half of $H_m$. Throughout this section we measure the balance tolerance relative to the target half-size $N/2$, so that an $\varepsilon$-balanced 2-partition $\{A, A^c\}$ of $H_m$ requires $|A|, |A^c| \in [(\tfrac{1}{2} - \varepsilon)N,\, (\tfrac{1}{2} + \varepsilon)N]$. Even appending the entire highway corridor $L \cup \{v_T\}$, which has $|V(L)| + 1 = 2(2m+1) + 1$ vertices, to a single region falls short of the lower endpoint of this interval:
\[
|V(R_i)| + |V(L)| + 1 \;<\; \Bigl[\bigl(\tfrac{1}{2} - \varepsilon\bigr)N - \bigl(2(2m+1) + 1\bigr)\Bigr] + \bigl(2(2m+1) + 1\bigr) \;=\; \bigl(\tfrac{1}{2} - \varepsilon\bigr)N.
\]
Hence even appending the entire highway corridor to a single region does not yield enough mass for a balanced half (and this holds for every $N$ and every $\varepsilon \geq 0$, with no large-$N$ side condition). It follows that every $\varepsilon$-balanced 2-partition of $H_m$ places at least a portion of two of the three regions $R_L, R_R, R_T$ on the same side of the cut.
 
Moreover, because $R_T$ contains a single entrance vertex $v_T$, a balanced cut can never split $R_T$ (as one of the produced regions would be contained entirely in $R_T$ and therefore have less than half the mass).
 
This means then that any $\varepsilon$-balanced 2-partition of $H_m$ must place at least a portion of $R_L$ and $R_R$ on the same side. In other words, any partition cannot have $R_L$ completely intact on one side and $R_R$ completely intact on the other.
\end{remark}

\section{Combinatorial Obstruction: Two Rungs Forbid a Balanced Cut}
 
In this section, we prove the following lemma:
 
\begin{lemma}\label{lem:two-rungs-bad}
If $T$ is a spanning tree of $H_m$ such that $|\mathcal{R}_T^L| \geq 2$ and $|\mathcal{R}_T^R| \geq 2$, then there is no edge $e \in T$ whose removal splits $T$ into two $\varepsilon$-balanced components.
\end{lemma}
 
\begin{proof}
Suppose, for contradiction, that the removal of some edge $e \in T$ splits $T$ into two components $T_1$ and $T_2$ that form an $\varepsilon$-balanced 2-partition of $V(H_m)$.
 
Note that from the previous remark at least one of $R_L$ and $R_R$ is split between $T_1$ and $T_2$ because both cannot be intact. Without loss of generality, assume $R_L$ is split into $a_1 =  V(R_L) \cap T_1$ and $a_2 = V(R_L) \cap T_2$ where $a_1$ and $a_2$ are non-empty and $a_1 \cup a_2 = V(R_L)$. 
 
Now, note that because $T_1$ and $T_2$ are balanced partitions, they both traverse the highway from $R_L$ to other regions. The highway is two vertices thick. Clearly $T_1$ and $T_2$ cannot both enter the highway from the same vertex because they are not connected. Without loss of generality let $T_1$ enter via the bottom vertex and $T_2$ enter via the top vertex. 
 
\begin{figure}[H]
    \centering
    \includegraphics[width=0.60\textwidth, , trim=5 5 5 7pt, clip]{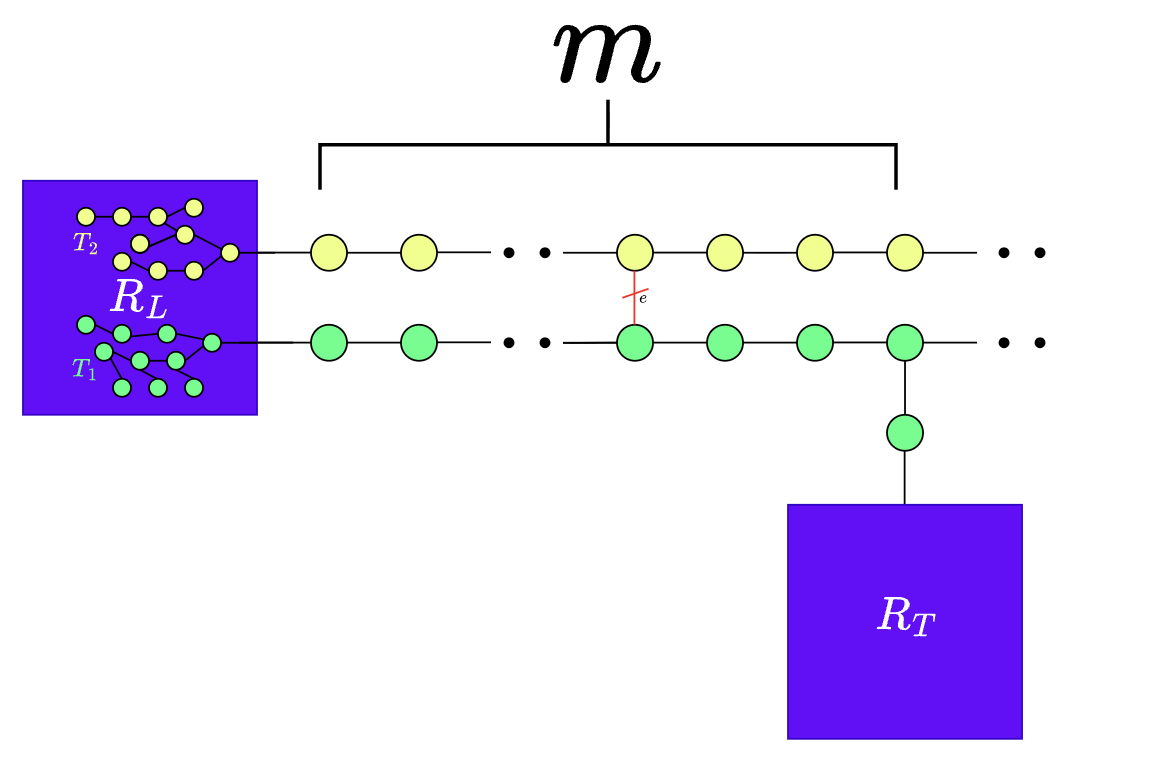}
    \caption{$T$ is drawn such that $T_1$ and $T_2$ represent a balanced partition. Note that there is only one vertical rung between the two rails and that the rung is $e$, the cut-edge that when removed produces the balanced partition.}
    \label{fig:balanced}
\end{figure}
 
Similarly, because they both traverse the highway, $T_1$ and $T_2$ must both reach the middle of the ladder. $T_1$ and $T_2$ are vertex disjoint, so the only way they can both reach the middle of the ladder is if each tree follows its horizontal rail as in Figure~\ref{fig:balanced}. That is, if there exists a rung in $T_1$ or $T_2$, then that tree would control both vertices in a column along the highway corridor, blocking the other tree from reaching the middle of the ladder (as in Figure~\ref{fig:blocked}). 
\begin{figure}[H]
    \centering
    \includegraphics[width=0.6\textwidth]{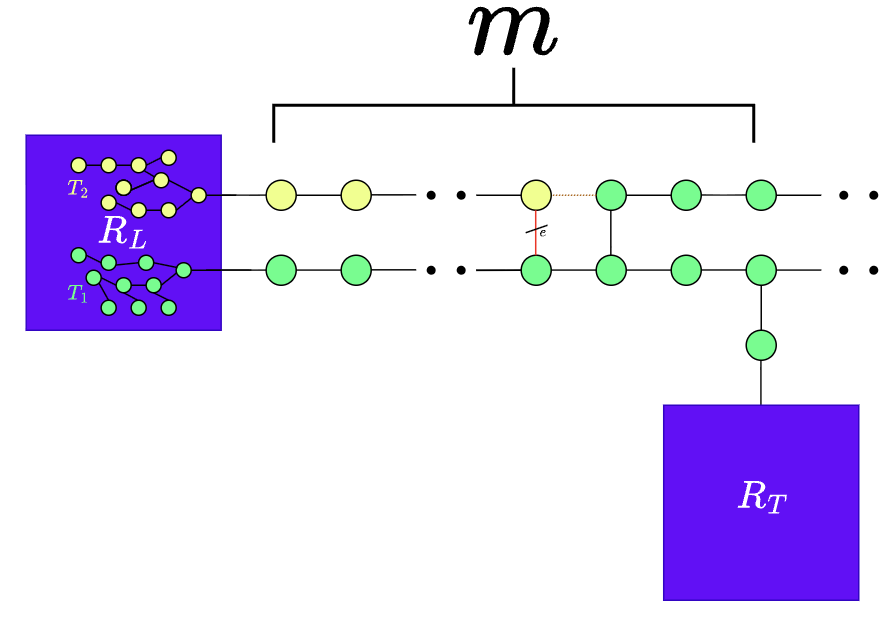}
    \caption{In this figure, there are two rungs in $T$ and consequently a vertical rung in $T_1$. As a result, $T_1$ contains both vertices in the $(m-2)$th column of the left portion of the highway. This means that $T_2$ is blocked from traversing the highway because the existence of the brown dashed edge would necessitate the connection of $T_1$ and $T_2$. As a result, no $\varepsilon$ balanced split is admitted.}
    \label{fig:blocked}
\end{figure}

However, this implies that $T_1 \cup T_2$ contains zero vertical rungs, which implies $T_1\cup T_2 \cup  e$ contains at most $1$ vertical rung. But this is a contradiction as $T_1 \cup T_2 \cup e = T$ and $|\mathcal{R}_T^L| \geq 2$.
 
The same contradiction is reached if $T_1$ and $T_2$ fall on the right side of the ladder by symmetry. So we conclude that $T_1$ and $T_2$ are not $\varepsilon$ balanced. 
\end{proof}
 
\section{Lemma: Uniform Spanning Trees Have Many Rungs in Each Half With High Probability}
\label{sec:lemma2}
 
Lemma~\ref{lem:two-rungs-bad} requires at least two rungs in each half of the ladder. We now show that a uniform spanning tree of $H_m$ satisfies this two-sided condition with high probability.
 
\begin{lemma}\label{lem:few-rungs-rare}
Let $T$ be a uniform spanning tree of $H_m$. Then
\[
\mathbb{P}\bigl[|\mathcal{R}_T^L| \leq 1 \text{ or } |\mathcal{R}_T^R| \leq 1\bigr] \;\leq\; 2\left(\tfrac{8}{9}\right)^{\lfloor m/2 \rfloor}.
\]
\end{lemma}
 
\begin{proof}
We start by noting that both halves of the $H_m$ highway in Figure~\ref{fig:highway} are composed of $\lfloor m/2 \rfloor$ consecutive disjoint 4-cycles (with each cycle formed by two rungs edges and two rail edges). We can label these cycles from left-to-right $C_{1} . . . C_{\lfloor \frac{m}{2}\rfloor}$ for each half of the ladder. 
 
Now use Algorithm~\ref{alg:ust-resistance} to sample $T$. As noted in the preliminaries, the probability $T$ is sampled is independent of the order in which edges are processed, so we can process all $2 \times 2\lfloor m/2 \rfloor$ rungs from left to right before any other edges of $H_m$.
 
Consider any of the $4$-cycles, say $C_i$, with its two rungs $r_i^{\top}$ and $r_i^{\bot}$ (which lie in distinct columns and therefore share no vertex). By Definition~\ref{def:highway}, every vertex incident to a rung of $C_i$ has degree $3$ in $H_m$. By Lemma~\ref{lem:degree-bound},
\[
R_{r_i^{\top}} \;\geq\; \frac{1}{3}, \qquad R_{r_i^{\bot}} \;\geq\; \frac{1}{3}.
\]
We sample $T$ with Algorithm~\ref{alg:ust-resistance}, processing the rungs cycle by cycle in the order $C_1, C_2, \dots$, and within each cycle processing $r_i^{\top}$ before $r_i^{\bot}$. Suppose the rungs of $C_1, \dots, C_{i-1}$ have already been processed, leaving a contracted/deleted graph $H'$. Because the cycles are pairwise vertex-disjoint, the endpoints of $r_i^{\top}$ and $r_i^{\bot}$ still have degree $3$ in $H'$, so by Lemmas~\ref{lem:resistance-tree} and~\ref{lem:degree-bound} the conditional probability that $r_i^{\top}$ is contracted into $T$ is its effective resistance in $H'$, which is at least $\tfrac13$. Now condition further on $r_i^{\top} \in T$, i.e.\ contract $r_i^{\top}$. This does change effective resistances elsewhere in the network; what matters is only that, since $r_i^{\bot}$ is vertex-disjoint from $r_i^{\top}$, its two endpoints retain degree $3$ in the contracted graph, so Lemma~\ref{lem:degree-bound} applies there and gives $R_{r_i^{\bot}} \geq \tfrac13$ in that graph. Hence the conditional probability that $r_i^{\bot} \in T$ is also at least $\tfrac13$, and
\[
\mathbb{P}\bigl[r_i^{\top} \in T \;\text{and}\; r_i^{\bot} \in T \,\big|\, \text{history through } C_{i-1}\bigr] \;\geq\; \frac{1}{9}.
\]
 
Equivalently, whatever the outcomes for the earlier cycles, the conditional probability that $C_i$ contributes at most one rung to $T$ is
\[
\mathbb{P}\bigl[|\{r_i^{\top}, r_i^{\bot}\} \cap T| \leq 1 \,\big|\, \text{history through } C_{i-1}\bigr] \;\leq\; \frac{8}{9}.
\]
 
If $|\mathcal{R}_T^L| \leq 1$, then every one of the $\lfloor m/2 \rfloor$ disjoint cycles in the left half contributes at most one rung. Multiplying the conditional bounds above along the processing order (rather than appealing to independence, which fails for spanning-tree edge events), the chain rule gives:
\[
\mathbb{P}\bigl[|\mathcal{R}_T^L| \leq 1\bigr] \;\leq\; \prod_{i=1}^{\lfloor m/2 \rfloor} \frac{8}{9} \;=\; \left(\frac{8}{9}\right)^{\lfloor m/2 \rfloor}.
\]
An identical argument applies to the right half of the ladder, so using a union bound we get that:
\[
\mathbb{P}\bigl[|\mathcal{R}_T^L| \leq 1 \;\text{or}\; |\mathcal{R}_T^R| \leq 1\bigr] \;\leq\; 2\left(\frac{8}{9}\right)^{\lfloor m/2\rfloor}. \qquad \qedhere
\]
\end{proof}

Combining Lemmas~\ref{lem:two-rungs-bad} and \ref{lem:few-rungs-rare} yields:
 
\begin{corollary}
\label{cor:no-balanced-cut}
For a uniform spanning tree $T$ of $H_m$,
\[
\mathbb{P}\bigl[\,T \text{ admits an $\varepsilon$-balanced cut}\,\bigr] \;\leq\; 2\left(\frac{8}{9}\right)^{\lfloor m/2\rfloor} .
\]
\end{corollary}
 
\begin{proof}
By Lemma~\ref{lem:two-rungs-bad}, if $T$ admits an $\varepsilon$-balanced cut then $|\mathcal{R}_T^L| \leq 1$ or $|\mathcal{R}_T^R| \leq 1$. Apply Lemma~\ref{lem:few-rungs-rare}.
\end{proof}
\section{The Adversarial 3-Partition}
\label{sec:partition}
 
We now embed the highway graph as the merge of two adjacent districts in a 3-partition of a planar square grid graph.
 
\begin{theorem}
\label{thm:main}
There exists an infinite family $\{(G_n, \mathcal{P}_n)\}_{n=1}^{\infty}$ of 
 planar square grid graphs with $\varepsilon$-balanced $3$-partitions, for any fixed balance tolerance $\varepsilon \in [0, \tfrac{1}{12})$, such that, writing $m = m(n) = \Theta(n)$ for the half-length of the ladder in the construction, the expected 
number of steps for \textsf{ReCom} to leave $\mathcal{P}_n$ is at least $exp(n)$
\end{theorem}
\begin{proof}
We construct $G_n$ and the partition $\mathcal{P}_n = \{R_1^n, R_2^n, R_3^n\}$ as illustrated in Figure~\ref{fig:partition}. The picture is a vertically-symmetric grid graph divided into four rectangular columns of widths roughly proportional to $\frac{1}{4}, \frac{1}{6}, \frac{1}{6}, \frac{1}{4}$ of the total mass (with $\varepsilon$-corrections), arranged so that:
\begin{itemize}
    \item $R_2^n$ (green) consists of the two narrow middle columns as well as a long thin tendril wrapping and dividing $R_1^n$ from $R_3^n$. 
    \item $R_1^n$ (brown) and $R_3^n$ (blue) each consist of one wide outer column connected via a thin bridge to an arm of size $\frac{1}{12} - \varepsilon$ that wraps one of the columns of $R_2^n$.
\end{itemize}
 
\begin{figure}[h]
    \centering
    \includegraphics[width=0.85\textwidth]{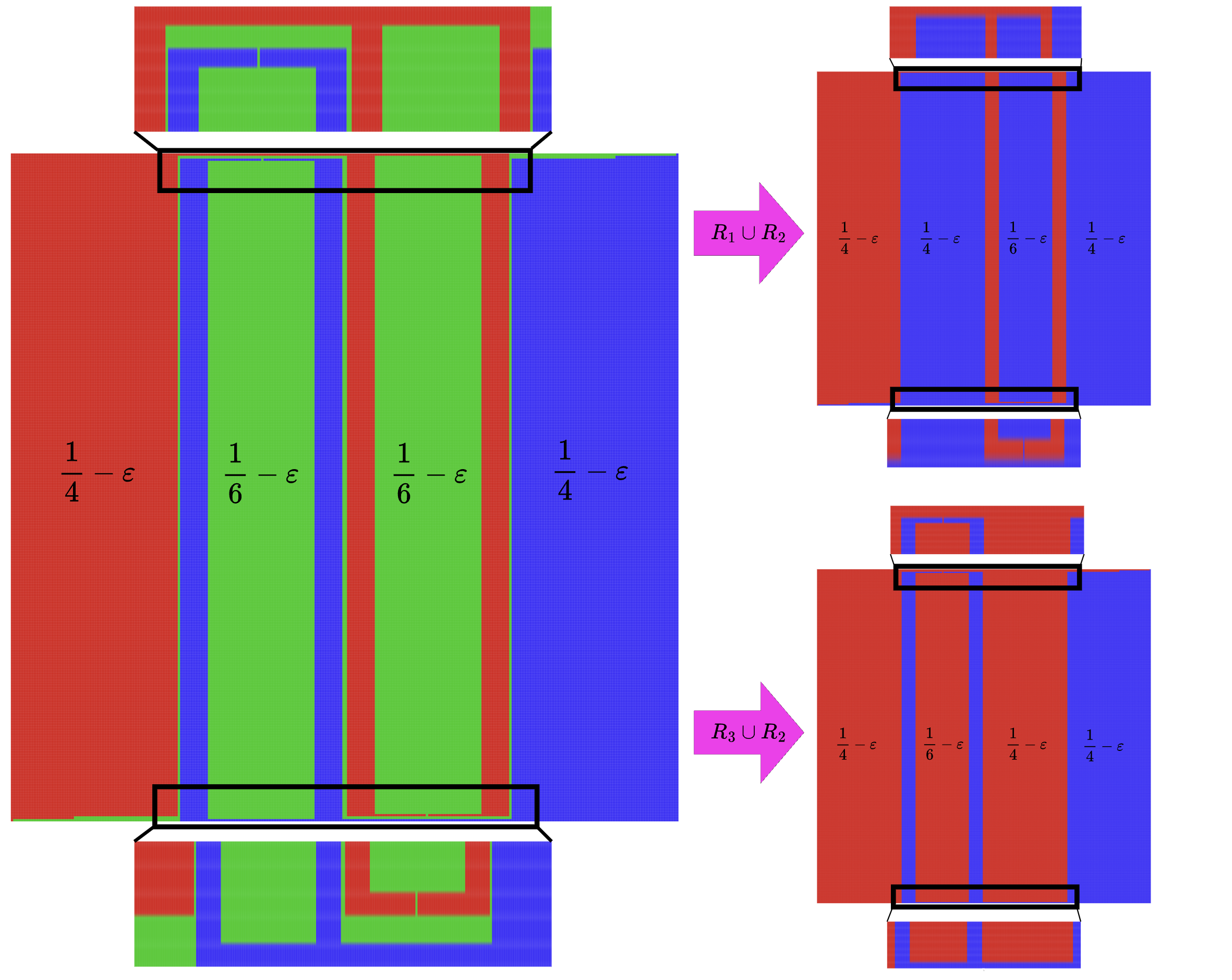}
    \caption{On the left, the 3-partition $\mathcal{P}_n = \{R_1^n, R_2^n, R_3^n\}$ on the $264 \times 264$ grid graph $G_n$, colored in brown, green, and blue respectively. Each region has $23232$ vertices, $\frac{1}{3}$ of the mass. On the right, the highway graphs created by the labeled merges. The black rectangles inscribe regions that are shown scaled and zoomed-in; these zoomed views are included to make it easier to identify the ``ladder'' portions of the graphs.}
    \label{fig:partition}
\end{figure}

By construction, there are only two possible merges in this graph because $R_1^n$ is nowhere adjacent to $R_3^n$. Moreover, both merges are highway graphs. To see this, consider the merge $R_1^n \cup R_2^n$. The merge consists of two wide outer regions of mass $\frac{1}{4} - \varepsilon$ and one wide inner region of mass $\frac{1}{6} - \varepsilon$. The two wide outer regions are joined by a $2 \times (2m+1)$ ladder, which is also connected to the inner region by a single vertex. If we denote the wide outer regions on the left and right with $R_L$ and $R_R$, and the inner region with $R_T$, the result is the highway graph $H_{m}$ of Definition~\ref{def:highway}. To see this, note that each of the three regions has at most $\tfrac13 N$ vertices, while the highway corridor contributes only $2(2m+1)+1 = \Theta(n)$ vertices --- an $O(1/n)$ fraction of the total mass $N = \Theta(n^2)$. Hence, for large $n$, each region together with the corridor still satisfies the size bound of Definition~\ref{def:highway} (since $\tfrac13 N + \Theta(n) < (\tfrac12 - \varepsilon)N$ once $N$ is large, using $\varepsilon < \tfrac{1}{12}$).

Similarly, because the initial construction is symmetric, the merge $R_2^n \cup R_3^n$ produces a symmetric highway graph.
 
Now note that \textsf{Reversible ReCom}'s transition function picks one of the three unordered pairs of pieces uniformly at random. With probability $1/3$ it picks $\{R_1^n, R_3^n\}$, which is non-adjacent, and stalls. With probability $2/3$ it picks an adjacent pair, in which case Corollary~\ref{cor:no-balanced-cut} bounds the conditional probability of finding a balanced cut by $2 \left(\tfrac{8}{9}\right)^{\lfloor m/2 \rfloor}$. Therefore the probability that any step succeeds in proposing a non-trivial move is at most
\[
\frac{2}{3} \cdot 2\left(\tfrac{8}{9}\right)^{\lfloor m/2 \rfloor} \;=\; \frac{4}{3} \left(\tfrac{8}{9}\right)^{\lfloor m/2 \rfloor}.
\]
Even ignoring the further reduction from the reversibility correction $1/E(P_i', P_j')$, this is exponentially small in $m$, and the expected number of outer iterations until the state changes is the reciprocal of this probability, namely $\Omega\!\left( \left(\tfrac{9}{8}\right)^{\lfloor m/2 \rfloor}\right)$.
\end{proof}
\section{Discussion}
 
Theorem~\ref{thm:main} establishes that, in the worst case over starting states, the \textsf{ReCom} transition function does not run in expected polynomial time with respect to $m$. In the above construction the variable $m$ grows with $n$, the width of the grid graph, so we can say that the \textsf{ReCom} transition function does not run in polynomial time with respect to the dimensions of the grid graph. 
 
This affirmatively answers the first half of Problem 2.8 \cite{aimpl}. The chain \emph{can} get stuck in configurations where a uniform random spanning tree of the merge admits a balanced cut only with exponentially small probability. The second half of Problem 2.8 --- which asks whether such ``bad'' configurations are exponentially unlikely under the chain's stationary distribution, or alternatively, exponentially unlikely to be reached from a typical starting state --- is left open.

We note two features of our construction that may bear on future work. First, the bad 3-partitions $\mathcal{P}_n$ are themselves moderately compact in the cut-edge sense, so compactness alone does not rule out adversarial configurations, even though the spanning-tree distribution favors compact partitions \cite{cannon2022}. Second, our graphs $G_n$ are grid graphs, which means that \textsf{ReCom} difficulty is not just an artifact of pathological planar geometry.
 
A natural next step would be to bound the probability that the \textsf{ReCom} markov chain, started from the stationary distribution of an $n \times n$ grid, ever reaches a configuration of the type produced by Theorem~\ref{thm:main}.

Finally, while we have found a family of adversarial configurations for \textsf{ReCom}, it should be noted that there are other related MCMC algorithms used for sampling districts which would not have to rediscover the thin path through the highway graph and would hence not be obstructed by this construction (e.g. the Balanced Up-Down Walk\cite{akitaya2026balancedupdownwalk}). 

\section*{Acknowledgements}
The author extends sincere thanks to Jamie Tucker-Foltz for introducing them to this problem and for providing invaluable mentorship and guidance throughout the preparation of this result.
\newpage
\appendix
\section{Reversible ReCom Pseudocode}
\label{app:recom}
 
\begin{algorithm}[h]
\caption{Reversible ReCom (one step) \cite{cannon2022}.}
\label{alg:recom}
\begin{algorithmic}[1]
\Require A connected planar graph $G$ on $kn$ vertices, an $\varepsilon$-balanced $k$-partition $\mathcal{P} = \{P_1, \dots, P_k\}$.
\State Select uniformly at random an unordered pair $\{P_i, P_j\}$.
\If{$P_i$ and $P_j$ are not adjacent in $G$}
    \State \Return $\mathcal{P}$
\EndIf
\State Sample a uniform random spanning tree $T$ of $G[P_i \cup P_j]$.
\State Run a depth-first search on $T$ to find an edge $e$ such that the two components of $T - e$ each have size $n$ (within $\varepsilon$ tolerance).
\If{no such $e$ exists}
    \State \Return $\mathcal{P}$
\EndIf
\State Label the two components of $T - e$ as $P_i', P_j'$.
\State Let $E(P_i', P_j')$ be the number of edges of $G$ with one endpoint in $P_i'$ and one in $P_j'$.
\State With probability $1/E(P_i', P_j')$, replace $\{P_i, P_j\}$ in $\mathcal{P}$ with $\{P_i', P_j'\}$.
\State \Return $\mathcal{P}$
\end{algorithmic}
\end{algorithm}

\bibliographystyle{plain}
\bibliography{refs}
 
\end{document}